\documentclass[10pt]{article}
\usepackage[latin1]{inputenc}
\usepackage[dvipsnames]{xcolor}
\usepackage[T1]{fontenc}
\usepackage{amsmath}
\usepackage{amsthm}
\usepackage{amssymb}
\usepackage{commath}
\usepackage{multicol}
\usepackage{xparse}
\NewDocumentCommand{\ceil}{s O{} m}{%
  \IfBooleanTF{#1} 
    {\left\lceil#3\right\rceil} 
    {#2\lceil#3#2\rceil} 
}
\usepackage{enumerate}
\usepackage[affil-it]{authblk}

\usepackage[margin=0.5in]{geometry}

\newtheorem{lemma}{Lemma}

\makeindex
\title{On the paper \lq\lq Quantum theory cannot consistently describe the use of itself''}

\author{Antonio Bernal%
  \thanks{Electronic address: \texttt{abernal@ub.edu}\\
Supported by project FIS2013-41757-P  }}
\affil{Department de Matemàtiques i Informàtica. Universitat de Barcelona}

\begin{document}
\maketitle

\begin{abstract}
In the paper \lq\lq Quantum theory cannot consistently describe the use of itself'' by D. Frauchiger and R. Renner an atempt is made at proving a \lq\lq no-go theorem'' that states that either quantum theory cannot be universally applied, even to macroscopic systems, or some very intuitive properties concerning recursive reasoning and uniquenes of physical values must be false.

In this paper, we give a concise description of the paper's result, and expose a detail in the proof.
\end{abstract}

\begin{multicols}{2}

\section*{Introduction}

This document deals on the paper \cite{frauchigerrenner}, by D. Frauchiger and R. Renner, on the limits of applicability of quantum mechanics. 

The paper strives to prove a \lq\lq no-go'' theorem that states that either quantum theory cannot be applied to macroscopic or \lq\lq reasoning'' systems or some other very intuitive assumptions (to be described later) must be false. 

To this end, the authors describe a \lq\lq Gedankenexperiment'' in which several agents interact with certain physical systems and apply quantum theory and ordinary reasoning, to obtain a contradictory result.

After the publication of \cite{frauchigerrenner}, there has been a number of papers stating several observations on the interpretation of the paper, some of them stating that the result is plainly wrong. See, for example,  \cite{paperadicional1,paperadicional2, paperadicional3,paperadicional4,paperadicional5,paperadicional6}. See also \cite{bub} for a paper aiming to clarify, and defend, the conclusions of \cite{frauchigerrenner}.

The aim of this paper is to describe the reasoning of \cite{frauchigerrenner} as sucinctly and tersely as possible, and to expose what the author thinks is a detail that renders the proof of the result of \cite{frauchigerrenner} incomplete.

The structure of this paper is as follows. In section \ref{descriptionofproblem}, the framework and methods  of \cite{frauchigerrenner} are explained concisely. In section \ref{theproof}, we give the details of the \lq\lq no-go'' theorem. In section \ref{theproblem}, we point out where a problem in the proof lies. In section \ref{basicobservation}, a little interpretation is given.

\section{The \lq\lq Gedankenexperiment'' and the \lq\lq no-go'' theorem}\label{descriptionofproblem}

To simplify the notation, the unitary operators that describe the evolution of the state of an isolated system from one instant of time to another will be implied. For example, suppose that the quantum system $X$ is isolated during a time interval $I$, and that it is in the state $|\Phi\rangle$ in time $t_1\in I$. Then, if $t_2\neq t_1$ is also in $I$, we will say that $X$ is also in the state $|\Phi\rangle$ in time $t_2$, instead of saying that it is in the state $U_{t_1\to t_2}|\Phi\rangle$, or some similar notation.

In the initial setting, we consider four agents named $\bar F$, $\bar W$, $F$ and $W$. The notation $W$ stands for Wigner and $F$ for Wigner's friend, refering to the well known Wigner's friend paradox \cite{wignersfriend}.

All four agents are aware of the protocol to be described and are able to apply quantum theory and make logical deductions.

Agent $\bar F$ is inside a lab $\bar L$, which consists on $\bar F$ itself, several experimental devices $\bar D$  and a random generator $R$ with two outcomes: \emph{heads}, with probability $1/3$ and \emph{tails}, with probabilty $2/3$. We can think about $R$ as a quantum system with a twodimensional state space, with a basis ${\cal B}=\{|\textit{heads}\rangle_R, |\textit{tiles}\rangle_R\}$, which is in an initial state
\begin{equation}\label{inistateofr}
	|\textit{init}\rangle_R=\frac{1}{\sqrt{3}}|\textit{heads}\rangle_R+ \sqrt{\frac{2}{3}}|\textit{tails}\rangle_R,
\end{equation}

Thus $\bar L=\bar F\otimes\bar D\otimes R$.

On the other hand, agent $F$, together with a set of devices $D$ form what we provisionally call lab $L_0$, $L_0=F\otimes D$. Agents $W$ and $\bar W$ are outside of $\bar L$ and $L_0$ and have no inner information on the outputs of measurements and preparations that take place inside the labs. They are aware of the protocol that is being followed.

Time is labelled in rounds $n$ and a continuous number $x$, as in $t=n\colon x$. 

In the beginning of round $n$, $t=n\colon 00$ a spin$-1/2$ particle $S$ enters in $\bar L$ and $\bar F$ trigers the random generator $R$. That is, $\bar F$ measures $R$ in the basis $\cal{B}$. Let $r$ be its output. If $r=\textit{heads}$, $\bar F$ prepares $S$ in the state $|\downarrow\rangle_S$ with the device $\bar D$, if $r=\textit{tails}$, $\bar F$ prepares $S$ in $|\rightarrow\rangle_S=(1/\sqrt{2})(|\uparrow\rangle_S + |\downarrow\rangle_S)$, time is now $t=n\colon 01$. Agent $\bar F$ doesn't publish neither the outcome of $R$ nor the state in which she has prepared the system $S$.

The laboratory $\bar L$ with $\bar F$, $\bar D$ and the random generator $R$ is let alone and isolated and $S$ enters into lab $L_0$, say in $t=n:05$. We call $|\bar h\rangle_{\bar L}$ and $|\bar t\rangle_{\bar L}$ the states of lab $\bar L$ when the preparation of $S$ is over and $S$ has exited $\bar L$, in the cases where $r=\textit{heads}$ and $r=\textit{tails}$ respectively.

Let's call $L=L_0\otimes S$ the second lab, which consists on agent $F$, the devices $D$ on it and the spin $S$. The system $L$ is isolated from $t=n:05$, when $S$ enters in $L_0$, to the moment when $W$ measures it, as we will later describe.

On time $t=n\colon 10$, agent $F$ measures $S$ on the basis $\{|\uparrow\rangle_S,|\downarrow\rangle_S\}$. Agent $F$ doesn't publish the result of this measurement. We call $|z=+1/2\rangle_L$ and $|z=-1/2\rangle_L$ the states of lab $L$ when the result of such a measure is respectively $|\uparrow\rangle_S$ and $|\downarrow\rangle_S$.

On time $t=n\colon 20$, agent $\bar W$ measures $\bar L$ on any basis containing the vector 
\[
|\bar{ok}\rangle_{\bar L} = \frac{1}{\sqrt{2}}(|\bar h\rangle_{\bar L} - |\bar t\rangle_{\bar L}).
\]
If the vector $|\bar{ok}\rangle_{\bar L}$ is selected in the measure, $\bar W$ declares publicy $\bar w=\bar{ok}$. Otherwise, he declares publicly $\bar w=\bar{fail}$.

On time $t=n\colon 30$, agent $W$ measures $L$ on any basis that contains the vector
\[
|ok\rangle_L = \frac{1}{\sqrt{2}}\left(|z=-\frac{1}{2}\rangle_L - |z=+\frac{1}{2}\rangle_L\right).
\]
If the result of the measure is $|ok\rangle_L$, $W$ declares publicly $w=ok$, otherwise, he declares publicy $w=fail$.

If $\bar w = \bar{ok}$ and $w=ok$, the experiment halts in time $t=n\colon 40$. Otherwise, the round $n$ is increased and the experiment starts again at time $t=n+1\colon 00$, resetting the devices and the random generator to their initial states, discarding the spin $S$ from $L$ and and entering another spin $S$ in $\bar L$.


The reasoning in \cite{frauchigerrenner} uses three assumptions Q, C and S. 

\emph{Assumption Q} assures that agents $F$, $\bar F$, $W$ and $\bar W$ can use quantum mechanics to analyze any physical system, even in the case where such system is macroscopic and contains devices or macroscopic agents that can make logical deductions and apply quantum mechanics themselsves to other systems. In short, assumption Q means that quantum theory can be applied to any system.

\emph{Assumtion C} means that if agent $A$ knows the conclusions of agent $A'$ and both agents use the same theory to get their conclusions, then A can adopt $A'$'s conclusions as his own. Assumption C is cognitive consistency.

\emph{Assumption S} means that the result of any single quantum measurement is unique. 

For completeness, we cite literally \cite{frauchigerrenner}:

\emph{Assumption Q. }Suppose agent $A$ has established that \lq\lq System $S$ is in state $|\psi\rangle_S$ at time $t_0$.''

Suppose furthermore that agent $A$ knows that \lq\lq the value $x$ is obtained by a measurement of $S$ with respect to the family $\{\pi^{t_0}_x\}_{x\in \chi}$ of Heisenberg operators relative to time $t_0$, wich is completed at time $t$.

If $\langle\psi|\pi^{t_0}_\xi|\psi\rangle=1$, for some $\xi\in\chi$, the agent $A$ can conclude that \lq\lq I am certain that $x=\xi$ at time $t$.''

\emph{Assumption C. }Suppose that agent $A$ has established that \lq\lq I am certain that agent $A'$, upon reasoning with the same theory as the one I am using, is certain that $x=\xi$ at time $t$.'' Then, agent $A$ can conclude \lq\lq I am certain that $x=\xi$ at time $t$.''

\emph{Assumption S. }Suppose agent $A$ has established that \lq\lq I am certain that $x=\xi$ at time $t$''. Then agent $A$ must necessarily deny that \lq\lq I am certain that $x\neq \xi$ at time $t$.''

\section{The proof}\label{theproof}

After having prepared the spin $S$, agent $\bar F$ sends it to $F$ without revealing to her the result of the random trial $r$ and the state she has prepared $S$ in. This happens in $t=n:05$.

To  agents $\bar W$ and $W$, (and to $F$ just an instant prior to the measurement of $S$ in $t=n\colon 10$) the system $\bar L\otimes S$ is in the state:

\begin{align}\label{statelbartensors}
	|\Psi_1\rangle_{\bar L\otimes S} &= \frac{1}{\sqrt{3}}|\bar h\rangle_{\bar L}|\downarrow\rangle_S + \sqrt{\frac{2}{3}} |\bar t\rangle_{\bar L}|\rightarrow\rangle_S\nonumber\\
	&= \frac{1}{\sqrt{3}}|\bar h\rangle_{\bar L}|\downarrow\rangle_S + \sqrt{\frac{2}{3}}\frac{1}{\sqrt{2}} |\bar t\rangle_{\bar L}\left(|\uparrow\rangle_S + |\downarrow\rangle_S\right)\nonumber\\
	&= \frac{1}{\sqrt{3}}\left(|\bar h\rangle_{\bar L}|\downarrow\rangle_S + |\bar t\rangle_{\bar L}|\uparrow\rangle_S + |\bar t\rangle_{\bar L} |\downarrow\rangle_S\right).
\end{align}

Just after $F$ made the measurement of $S$ with respect to the basis $\{|\uparrow\rangle_S,|\downarrow\rangle_S\}$, the system $\bar L\otimes L$ is, to $W$ and $\bar W$, in the state:

\begin{align}\label{statelbartensorl}
|\Psi_2\rangle_{\bar L\otimes L} &= \frac{1}{\sqrt{3}}(|\bar h\rangle_{\bar L}|z=-\frac{1}{2}\rangle_L + |\bar t\rangle_{\bar L}|z=+\frac{1}{2}\rangle_L \nonumber\\
&+ |\bar t\rangle_{\bar L} |z=-\frac{1}{2}\rangle_L).
\end{align}

It can be directly checked that 
\begin{equation}\label{probabilityoneovertwelve}
	|\phantom{}_{\bar L}\langle \bar{ok}|\phantom{}_L\langle ok|\Psi_2\rangle_{\bar L\otimes L}|^2=\frac{1}{12},
\end{equation}
so we can be certain that there will be a round $n^*$ in which the experiment will halt at $t=n^*\colon 40$.

However, it is also direct that
\begin{equation}\label{okbarminusonehalfimpossible}
|\phantom{}_{\bar L}\langle \bar{ok}|\phantom{}_L\langle z=-\frac{1}{2}|\Psi_2\rangle_{\bar L\otimes L}|^2=0.
\end{equation}

By (\ref{probabilityoneovertwelve}), all the agents are certain that there will come a round $n^*$, such that $\bar W$ will publish $\bar w=\bar{ok}$ at time $t=n^*\colon 20$ and $W$ will publish $w=ok$ at time $t=n^*\colon 30$.

So, we suppose that we are in such a round $n^*$. 

Due to (\ref{okbarminusonehalfimpossible}), since $\bar w=\bar{ok}$, it is impossible that $F$ measured $z=-1/2$ at time $t=n^*\colon 10$, so it measured $z=+1/2$.

Then, it is impossible that $\bar F$ prepared $S$ in time $t=n^*\colon 01$ in the state $|\downarrow\rangle_S$. So, it turns out that $\bar F$ prepared $S$ in the state $|\rightarrow\rangle_S$, and that was the state in which $F$ received $S$ in time $t=n^*:05$.

Since $W$ knew in $t=n^*\colon 20$ the publication $\bar w=\bar{ok}$ by $\bar W$, all the above reasoning is available to $W$ in $t=n^*\colon 20$, so $W$ knew that $S$ was in the state $|\rightarrow\rangle_S$ in time $t=n^*\colon 09$. So, to $W$, the state of $L$ in time $t=n^*\colon 10$, after $F$ measured $S$ in the basis $\{|\uparrow\rangle_S,|\downarrow\rangle_S\}$, was

\begin{equation}\label{stateoflafterwknowsaboutwbarr}
\frac{1}{\sqrt{2	}}\left(\left|z=+\frac{1}{2}\right\rangle_L + \left|z=-\frac{1}{2}\right\rangle_L\right).
\end{equation}

Now lets point out a detail:

\begin{lemma}\label{lemmaconstantstate}
	The state of $L$ to $W$ is also (\ref{stateoflafterwknowsaboutwbarr}) in time $t=n^*\colon 30$.
\end{lemma}
\begin{proof}
Since $L$ is isolated from $t=n^*\colon 05$ onwards, until $W$ makes the measurement of $L$ in $t=n^*\colon 30$. 
\end{proof}

This state is orthogonal to $|ok\rangle_L$. So we know that it is impossible for $W$ to measure $w=ok$ in $t=n^*\colon 30$. A contradiction.

\section{The problem}\label{theproblem}

However, the state of $L$ for $W$ in time $t=n^*\colon 10$ is not (\ref{stateoflafterwknowsaboutwbarr}). 

Agent $\bar W$ has arrived to the conclusion that the state of $S$ to $F$ is $|\uparrow\rangle_S$ after the measure, so the state of $L$ was to him (\ref{stateoflafterwknowsaboutwbarr}) in $t=n^*\colon 09$, before the measurement, but it was
\begin{equation}\label{realstateofl}
	|z=+\frac{1}{2}\rangle_L
\end{equation}
in time $t=n^*\colon 10$, after the measurement.

Since $\bar W$ published the result of his measure $\bar w=\bar{ok}$ in $t=n^*\colon 20$, $W$ has the same knowledge and, for him, the state of $L$ in $t=n^*:10$ is not (\ref{stateoflafterwknowsaboutwbarr}), but (\ref{realstateofl}). And the state (\ref{realstateofl}) is not orthogonal to $|ok\rangle_L$, so $W$ can observe $w=ok$ at $t=n^*:30$. Thus, there are no contradictory conclusions. 

The key of the argument lies in Lemma \ref{lemmaconstantstate}. To $W$, the state of $L$ was (\ref{stateoflafterwknowsaboutwbarr}) in $t=n^*\colon 09$, $L$ was isolated from then up to $t=n^*\colon 30$. How is it possible that the state is (\ref{stateoflafterwknowsaboutwbarr}) in the begining and, without any interaction, it is (\ref{realstateofl}) in the end?

\section{A basic observation}\label{basicobservation}

Let's consider a simple example.

Suppose that Alice is preparing a qubit in a lab in either the state $|0\rangle$ or $|1\rangle$ with equal probabilities. Bob is waiting outside the lab in a room without any knowledge of Alice's preparation.

When she is finished, Alice places the qubit in a perfectly isolating container and handles it to Bob, without revealing the state in which the qubit has been prepared.

Bob is in his room with the container. To him, the state of the qubit is the maximally mixed state

\[
	|\Phi\rangle = \frac{1}{2}\mathbb{I}.
\]

Now suppose that Alice goes from her lab to Bob's room and tells Bob the state she has prepared the qubit in. Let's suppose that she lets Bob know that she prepared $|0\rangle$.

To Bob, the state of the qubit is now $|0\rangle$. So the qubit has been isolated from the Universe and without any internal dynamics, but by the change in Bob's knowledge, its state (to Bob) changed from $|\Phi\rangle$ to $|0\rangle$.

A way to interpret things could be that the perfect isolating container was not so isolating, after all. The qubit was entangled with Alice, and in talking with Alice, Bob made an indirect measurement of the qubit, so the evolution was not unitary. 

Another way of seeing it could be that the fact that a system, like the qubit, was isolated from external influences and had no internal dynamics, didn't make its state constant, since the state of a system depends, no only on the system, but also on the information that the observer has about it, and this information has changed.

In our case, the state of the system $L$ was (\ref{stateoflafterwknowsaboutwbarr}) in time $t=n^*:09$, but in $t=n^*:20$ $W$ gets some information from the register of $\bar W$, which is akin to an indirect measurement of $L$ by $W$. The evolution of $L$, therefore, ceases to be unitary and lemma \ref{lemmaconstantstate} does not apply.

\end{multicols}

\end{document}